\documentclass[journal]{IEEEtran}
\usepackage{dsfont,graphicx,color,epsf,epsfig,verbatim,amssymb,amsmath,amsthm}
\usepackage{latexsym,amsfonts,caption,array,cite,subfigure,multicol,multirow}
\newtheorem{lemma}{Lemma}[section]

\newtheorem{theorem}{Theorem}[section]
\newtheorem{deff}{Definition}[section]

\begin{document}
\title{Asymptotically Good Codes over\\ non-Abelian Groups}
\author{\IEEEauthorblockN{Aria G. Sahebi and S. Sandeep Pradhan\\ \thanks{This work was supported by NSF grants CCF-0915619 and CCF-1116021.}}
\IEEEauthorblockA{Department of Electrical Engineering and Computer Science,\\
University of Michigan, Ann Arbor, MI 48109, USA.\\
Email: \tt\small ariaghs@umich.edu, pradhanv@umich.edu}}

\markboth{Draft}
{Sahebi \MakeLowercase{\textit{et al.}}: Asymptotically Good Codes over non-Abelian Groups}
\maketitle

\begin{abstract}
In this paper, we show that good structured codes over non-Abelian groups do exist. Specifically, we construct codes over the smallest non-Abelian group $\mathds{D}_6$ and show that the performance of these codes is superior to the performance of Abelian group codes of the same alphabet size. This promises the possibility of using non-Abelian codes for multi-terminal settings where the structure of the code can be exploited to gain performance.
\end{abstract}


\IEEEpeerreviewmaketitle

\section{Introduction}
\IEEEPARstart{A}{lgebraically} structured codes are an important class of codes in coding/information theory and communications and evaluating the information-theoretic performance limits of such codes has been an area of significance \cite{ahlswede_alg_codes,sandeep_discus,forney_dynamics,fagnani_abelian,dinesh_dsc,dobrushin_group}. It is well-known that linear codes achieve the symmetric capacity of $q$-ary channels where $q$ is a prime \cite{elias}\cite{dobrushin_group}. Linear codes can also be used to compress a binary source losslessly down to its entropy \cite{korner_marton}. Optimality of linear codes for certain communication problems motivates the study of algebraic-structured codes including Abelian and non-Abelian group codes.

In \cite{korner_marton} it has been shown that for some multi-terminal communication settings, the average asymptotic performance of the ensemble of structured codes can be better than that of random codes. In recent years, such gains have been shown for a wide class of multi-termianl problems \cite{phiosof_zamir, dinesh_dsc, nazer_gastpar}. Thus, characterization of the information theoretic performance limits of these codes became important. However, the structure of the code restricts the encoder to abide by certain algebraic rules. This causes the performance of such codes to be inferior to random codes in some communication settings. Linear codes are highly structured and for some problems in information theory they cannot be optimal. Moreover, these codes can only be defined over alphabets of size a power of a prime.

Group codes are a generalization of linear codes which are algebraically structured and can be defined for any alphabet. These codes can outperform unstructured codes in certain communication problems \cite{dinesh_dsc}. Group codes were first studied by Slepian \cite{slepian_group} for the Gaussian channel. In \cite{ahlswede_group}, the capacity of group codes for certain classes of channels has been computed. Further results on the capacity of group codes were established in \cite{ahlswede_alg_codes,ahlswede_alg_codes2,abelianp2p_ieee}.\\
It has been shown in \cite{abelianp2p_ieee} that Abelian group codes do not achieve the capacity of arbitrary channels. It has also been conjectured by several authors that non-Abelian group codes are inferior to Abelian group codes \cite{Forney_hamming_distance_group_codes} \cite{Interlando_group_codes_are_bad} \cite{Massey_nonabelian_conformant}. This motivates a loosening of the structure of the code yet further.

In this work, we focus on the point-to-point channel coding problem. We define a class of structured codes which includes the class of group codes and has less structure compared to group codes. We evaluate the performance of such codes over the smallest non-Abelian group $\mathds{D}_6$ and show that these codes have a strictly better performance compared to Abelian group codes. We use a combination of algebraic and information-theoretic tools for this task. This observation broadens our view to structured codes for possible use in multi-terminal settings.\\
The paper is organized as follows: In Section \ref{section:preliminaries}, we introduce
our notation and develop the required background. In Section \ref{section:ensemble} we define the ensemble of codes. We analyze the performance of theses codes in Section \ref{section:error_analysis} where we solve an optimization problem and make several counting arguments. We compare the performance of the constructed codes to the performance of Abelian group codes in Section \ref{section:comparison} and we conclude in Section \ref{section:Conclusion}.

\section{Preliminaries}\label{section:preliminaries}
\subsubsection{Groups}
A group is a set $G$ equipped with a binary operation ``$\cdot$'' to form an algebraic structure. The group operation ``$\cdot$'' must satisfy the group axioms (closure, associativity, identity and invertibility). A group is called \emph{Abelian} if its operation is commutative and \emph{non-Abelian} otherwise.

\subsubsection{Group Codes}
Given a group $G$, a group code $\mathds{C}$ over $G$ with block length $n$ is any subgroup of $G^n$ \cite{forney_dynamics,algebra_bloch}. A shifted group code over $G$, $\mathds{C}+v$ is a translation of a group code $\mathds{C}$ by a fixed vector $v\in G^n$.

\subsubsection{Source and Channel Models}
We consider discrete memoryless and stationary channels used without feedback. We associate two finite sets $\mathcal{X}$ and $\mathcal{Y}$ with the channel as the channel input and output alphabets. These channels can be characterized by a conditional probability law $W(y|x)$ for $x\in \mathcal{X}$ and $y\in \mathcal{Y}$. The set $\mathcal{X}$ admits the structure of a finite Abelian group $G$ of the same size. The channel is specified by $(G,\mathcal{Y},W)$. Assuming a perfect source coding block applied prior to the channel coding, the source of information generates messages over the set $\{1,2,\ldots,M\}$ uniformly.

\subsubsection{Achievability and Capacity}
A transmission system with parameters $(n,M,\tau)$ for reliable communication over a given channel $(G,\mathcal{Y},W)$ consists of an encoding mapping and a decoding mapping
\begin{align*}
&e:\{1,2,\ldots,M\}\rightarrow G^n\\
&f:\mathcal{Y}^n\rightarrow\{1,2,\ldots,M\}
\end{align*}
such that for all $m=1,2,\ldots,M$,
\begin{align*}
\frac{1}{M}\sum_{m=1}^{M}W^n\left(f(Y^n)\ne
m|X^n=e(m)\right)\le \tau
\end{align*}
Given a channel $(G,\mathcal{Y},W)$, the rate $R$ is said to be achievable if for all $\epsilon>0$ and for all sufficiently large $n$, there exists a transmission system for reliable communication with parameters $(n,M,\tau)$ such that
\begin{align*}
\frac{1}{n}\log M \ge R-\epsilon,\mbox{   }\tau\le \epsilon
\end{align*}
The capacity of the channel is defined as the supremum of the set of all achievable rates.
\subsubsection{Typicality}
Consider two random variables $X$
and $Y$ with joint probability density function $p_{X,Y}(x,y)$
over $\mathcal{X}\times\mathcal{Y}$. Let $n$ be an integer and
$\epsilon$ be a positive real number. The sequence pair $(x^n,y^n)$
belonging to $\mathcal{X}^n\times \mathcal{Y}^n$ is said to be
jointly $\epsilon$-typical with respect to $p_{X,Y}(x,y)$ if
\begin{align*}
\nonumber \forall a\in\mathcal{X},\mbox{   }\forall b\in\mathcal{Y}:
\left|\frac{1}{n}N\left(a,b|x^n,y^n\right)-p_{X,Y}(a,b)\right|\le
\frac{\epsilon}{|\mathcal{X}||\mathcal{Y}|}
\end{align*}
and none of the pairs $(a,b)$ with $p_{X,Y}(a,b)=0$ occurs in
$(x^n,y^n)$. Here, $N(a,b|x^n,y^n)$ counts the number of
occurrences of the pair $(a,b)$ in the sequence pair $(x^n,y^n)$.
We denote the set of all jointly $\epsilon$-typical sequence
pairs in $\mathcal{X}^n\times \mathcal{Y}^n$ by
$A_\epsilon^n(X,Y)$.\\
Given a sequence $x^n\in \mathcal{X}^n$, the set of conditionally
$\epsilon$-typical sequences $A_\epsilon^n(Y|x^n)$ is defined as
\begin{eqnarray}
A_\epsilon^n(Y|x^n)=\left\{y^n\in \mathcal{Y}^n\left| (x^n,y^n)\in
A_\epsilon^n(X,Y)\right.\right\}
\end{eqnarray}
\subsubsection{Dihedral Groups}
A dihedral group of order $2p$ is the group of symmetries of a regular $p$-gon, including reflections and rotations and any combination of these operations. A dihedral group can be represented as a quotient of a free group as follows:
\begin{align*}
D_{2p}=\langle x,y|x^p=1,y^2=1,xyxy=1\rangle
\end{align*}
Dihedral groups are among the simplest non-Abelian groups.
\subsubsection{Notation}In our notation, $O(\epsilon)$ is any function of $\epsilon$ such that $\lim_{\epsilon\rightarrow 0}O(\epsilon)=0$ and for a set $A$, $|A|$ denotes its size (cardinality).
\section{A Class of Structured Codes}\label{section:ensemble}
Based on Forney's \emph{analysis} of group codes \cite{forney_dynamics}, we \emph{construct} a class of structured codes which we call \emph{pseudo-group} codes.

In this note, we consider a class of pseudo-linear codes over $\mathds{D}_6$ and find a lower bound on the capacity of such codes.\\
Let $A_1,\cdots,A_n$ be subgroups of $G$. We propose a method to construct a code whose output at time $k$ forms the subgroup $A_k$ of $G$ and whose input is the subgroup $F_k$ of $A_k$:

\begin{itemize}
\item Choose the \emph{controllability index} of the code $\nu\in \mathbb{Z}^+$.
\item For each $k$, choose a normal series $F_{k,0}\lhd F_{k,1}\lhd \ldots \lhd F_{k,\nu}=F_k$ in $F_k$.
\item For each $k$, define \emph{granules} $\Gamma_{[k,k]}=F_{k,0}$ and for all $1\leq j\leq \nu$, define granules $\Gamma_{[k,k+j]}=F_{k,j}/F_{k,j-1}$.
\item For each $k$ and for all $0\leq j\leq \nu$ find $C_{[k,k+j]}$ such that $C_{[k,k+j]}\leq A_k^{[k,k+j]}$ and $C_{[k,k+j]}/C_{[k,k+j-1]}\cdot C_{[k+1,k+j]}\cong \Gamma_{[k,k+j]}$ where $A_k^{[k,k+j]}$ is the projection of $A_k^\infty$ onto the interval $[k,k+j]$. In other words, $C_{[k,k+j]}$ is isomorphically an extension of $C_{[k,k+j-1]}\cdot C_{[k+1,k+j]}$ by $\Gamma_{[k,k+j]}$ inside $A_k^{[k,k+j]}$. Note that at least one such extension exists (direct product) and note that the projections $C_{[k,k+j-1]}$ are defined recursively.
\item Choose the mappings $T_k:F_k\rightarrow C_{[k,k+\nu]}$ such that T is a transversal function of the subgroup $C_{[k,k+\nu-1]}\cdot C_{[k+1,k+\nu]}$ of $C_{[k,k+\nu]}$.
\item Given an information digit sequence $\ldots u_0u_1u_2\ldots$, $u_i\in \mathbb{Z}$, the output sequence is $\ldots T_0(u_0)T_1(u_1)T_2(u_2)\ldots$.
\end{itemize}

\begin{deff}
Any code constructed using the above construction algorithm is called a pseudo-linear code over $G$.
\end{deff}

\begin{theorem}
Any free group code is a pseudo-linear code.
\end{theorem}

\begin{proof}
It has been shown in \cite{forney_dynamics} that any free group code can be reconstructed using the above construction algorithm. This completes the proof.
\end{proof}
%

For Abelian groups, the definition of pseudo-group codes coincides with the definition of group codes but for non-Abelian groups this class is larger than the class of group codes; i.e. it includes all group codes as well as some non-group codes. In this paper, we use these codes for the smallest non-Abelian group $\mathds{D}_6$ and show that this loosening of the structure results in a better performance. The generalization of the analysis to dihedral groups $\mathds{D}_{2p}$ where $p$ is a prime is relatively straight forward.\\
The group $\mathds{D}_6$ with presentation $\mathds{D}_{6}=\langle x,y|x^3=1,y^2=1,xyxy=1\rangle$ can be characterized by a set $\{1,x,x^2,y,xy,x^2y\}$ with the following table of operations:
\begin{center}
\begin{tabular}{|c||c|c|c|c|c|c|}\hline
$\cdot$&$1$&$x$&$x^2$&$y$&$xy$&$x^2y$\\\hline\hline
$1$&$1$&$x$&$x^2$&$y$&$xy$&$x^2y$\\\hline
$x$&$x$&$x^2$&$1$&$xy$&$x^2y$&$y$\\\hline
$x^2$&$x^2$&$1$&$x$&$x^2y$&$y$&$xy$\\\hline
$y$&$y$&$x^2y$&$xy$&$1$&$x^2$&$x$\\\hline
$xy$&$xy$&$y$&$x^2y$&$x$&$1$&$x^2$\\\hline
$x^2y$&$x^2y$&$xy$&$y$&$x^2$&$x$&$1$\\\hline
\end{tabular}
\end{center}

Note that for two elements $g,h$ in $\mathds{D}_6$, $g\cdot h$ may not be equal to $h\cdot g$. In the following, we construct the ensemble of codes over $\mathds{D}_6$. Consider codes with input group $F_k=\mathds{D}_6$ and output group $A_k=\mathds{D}_6$ for all $k$. Let $\nu$ be a nonnegative integer number. We choose the following normal series:\\
\begin{align*}
F_{k,0}=1\lhd F_{k,1}=1\lhd \ldots F_{k,\nu-1}=1\lhd F_{k,\nu}=\mathds{D}_6=F_k
\end{align*}
(In the most general case we can have a chain isomorphic to:
\begin{align*}
1\lhd 1\lhd \ldots \lhd 1\lhd \mathds{Z}_2\lhd \ldots \lhd \mathds{Z}_2\lhd \mathds{D}_6\lhd \ldots \lhd\mathbb{D}_6
\end{align*}
but it can be shown that this chain can do no better.)\\
Since the input and output groups are not changing over time, the granules are time independent. Define $\Gamma_j=\Gamma_{[k,k+j]}$. Then,
\begin{align*}
&\Gamma_0=F_{k,0}=C_{[k,k]}=1\\
&\Gamma_1=F_{k,1}/F_{k,0}=1\\
&\vdots\\
&\Gamma_{\nu-1}=F_{k,\nu-1}/F_{k,\nu-2}=1\\
&\Gamma_{\nu}=F_{k,\nu}/F_{k,\nu-1}=\mathds{D}_6\\
\end{align*}
Next step is to find the projection of the code over finite intervals. Note that in this case we have $C_{[k,k+j]}\cong C_{[0,j]}$.
\begin{align*}
&C_{[0,0]}\cong\Gamma_0=1, C_{[0,0]}\le \mathds{D}_6^1\Rightarrow C_{[0,0]}=1\\
&C_{[1,1]}\cong C_{[0,0]}=1\\
&\Rightarrow C_{[1,1]}\cdot C_{[0,0]}=1\\
&C_{[0,1]}/C_{[0,0]}\cdot C_{[1,1]}\cong \Gamma_1=1, C_{[0,1]}\le \mathds{D}_6^2\Rightarrow C_{[1,1]}=1\\
&\vdots\\
&C_{[1,\nu-1]}\cong C_{[0,\nu-2]}=1\\
&\Rightarrow C_{[0,\nu-2]}\cdot C_{[1,\nu-1]}=1\\
&C_{[0,\nu-1]}/C_{[0,\nu-2]}\cdot C_{[1,\nu-1]}\cong \Gamma_{\nu-1}=1, C_{[0,\nu-1]}\le \mathbb{D}_6^\nu\\
&\Rightarrow C_{[0,\nu-1]}=1\\
&C_{[1,\nu]}\cong C_{[0,\nu-1]}=1\\
&\Rightarrow C_{[0,\nu-1]}\cdot C_{[1,\nu]}=1\\
&C_{[0,\nu]}/C_{[0,\nu-1]}\cdot C_{[1,\nu]}\cong \Gamma_{\nu}=\mathbb{D}_6, C_{[0,\nu]}\le \mathbb{D}_6^{\nu+1}\\
\end{align*}
Therefore,
\begin{align}\label{eqn:subcode}
C_{[0,\nu]}=\langle \textbf{g}^0,\textbf{h}^0|(\textbf{g}^0)^3=1,(\textbf{h}^0)^2=1,\textbf{g}^0\textbf{h}^0\textbf{g}^0\textbf{h}^0=1\rangle
\end{align}

where $\textbf{g}^0,\textbf{h}^0\in \mathds{D}_6^{\nu+1}$.\\
\begin{align*}
&(\textbf{g}^0)^3=1\Rightarrow \textbf{g}^0\in \{1,x,x^2\}^{\nu+1}\\
&(\textbf{h}^0)^2=1\Rightarrow \textbf{h}^0\in \{1,y,xy,x^2y\}^{\nu+1}\\
\end{align*}

It can be shown that if we take $\textbf{g}^0=g_{00}g_{01}\ldots g_{0\nu}$ and $\textbf{h}^0=h_{00}h_{01}\ldots g_{0\nu}$ where $g_{0i}$ and $h_{0i}$ are chosen jointly according to Table 1 then the third condition in Equation \ref{eqn:subcode} will also be satisfied. Therefore for any such $\textbf{g}^0$ and $\textbf{h}^0$, the group $C_{[0,\nu]}=\langle \textbf{g}^0,\textbf{h}^0|(\textbf{g}^0)^3=1,(\textbf{h}^0)^2=1,\textbf{g}^0\textbf{h}^0\textbf{g}^0\textbf{h}^0=1\rangle$ is a subgroup of $\mathds{D}_6$. As tough our input group is only restricted to be a subgroup of $\mathds{D}_6$ and not necessarily $\mathds{D}_6$ itself.\\
Similarly, We get $C_{[k,k+\nu]}=\langle \textbf{g}^k,\textbf{h}^k|(\textbf{g}^k)^3=1,(\textbf{h}^k)^2=1,\textbf{g}^k\textbf{h}^k\textbf{g}^k\textbf{h}^k=1\rangle$ where $\textbf{g}^k=g_{k0}g_{k1}\ldots g_{k\nu}$ and $\textbf{h}^k=h_{k0}h_{k1}\ldots g_{k\nu}$ and $g_{0i}$'s and $h_{0i}$'s are chosen according to Table 1.\\
Note that any element in $\mathds{D}_6$ can be uniquely written as $x^\alpha y^\beta$ for some $\alpha\in \mathds{Z}_3$ and $\beta\in \mathds{Z}_2$. Define the transversal functions as:\\
\begin{align*}
T_k(x^\alpha y^\beta)=(\textbf{g}^k)^\alpha (\textbf{h}^k)^\beta
\end{align*}
Let $\ldots u_0u_1u_2\ldots$ be the information digits where $u_i=x^{a_i}y^{b_i}$, then the output of the code is $\ldots c_0c_1c_2\ldots$ where $c_i=g_{k-\nu,\nu}^{a_{k-\nu}}h_{k-\nu,\nu}^{b_{k-\nu}}\cdots g_{k0}^{a_k}h_{k0}^{b_k}$.\\
Assume the input is fed circularly to the code (tail biting). We also add a dither to the code.\\

Here we give a summary of the resulting ensemble of codes. Each code in this ensemble has a rate of $R=\frac{k}{n}\log 6$.\\
\begin{itemize}
    \item For $i=1,\cdots,n$ and $j=1,\cdots,k$ choose $g_{ij}$ and $h_{ij}$ randomly according to Table 1. for $(i,j)\ne (i^\prime,j^\prime)$, $(g_{ij},h_{ij})$ is chosen independently from $(g_{i^\prime j^\prime},h_{i^\prime j^\prime})$.\\
    \item For $i=1,\cdots,n$, choose the dither $B_i$ uniformly randomly from $\mathds{D}_6$.
    \item Given the input sequence $u=(u_1,\cdots,u_k)$ where $u_i=x^{a_i}y^{b_i}$, $a_i\in\mathds{Z}_3$, $b_i\in\mathds{Z}_2$ for $i=1,\cdots,k$,  the output sequence is equal to $c=(c_1,\cdots,c_n)$ where
\begin{align}
\nonumber&c_1=g_{11}^{a_1}h_{11}^{b_1}g_{12}^{a_2}h_{12}^{b_2}\cdots g_{1k}^{a_k}h_{1k}^{b_k}\cdot B_1\\
\nonumber&c_2=g_{21}^{a_1}h_{21}^{b_1}g_{22}^{a_2}h_{22}^{b_2}\cdots g_{2k}^{a_k}h_{2k}^{b_k}\cdot B_2\\
\nonumber&\vdots\\
\label{eqn:encoder}&c_n=g_{n1}^{a_1}h_{n1}^{b_1}g_{n2}^{a_2}h_{n2}^{b_2}\cdots g_{nk}^{a_k}h_{nk}^{b_k}\cdot B_n
\end{align}
\end{itemize}
We denote this by $c=G(u)\cdot B$.\\

\begin{figure}[h!]\label{fig:table_gh}
  \centering
    \includegraphics[scale=1.3]{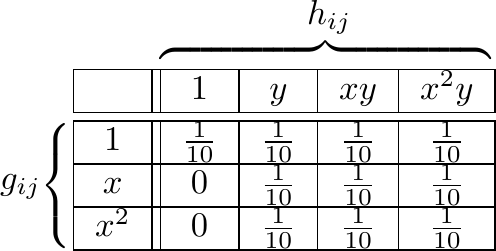}
    \caption{$g_{ij}$ is chosen from $\{1,x,x^2\}$ and $h_{ij}$ is chosen from $\{y,xy,x^2y\}$. The number in the table shows the joint probability of $(g_{ij},h_{ij})$ being picked.}
\end{figure}
We evaluate the performance of these codes using a random coding argument in the next section.
\section{Main Result}\label{section:error_analysis}
In this section we show the existence of good structured codes over the non-Abelian group $\mathds{D}_6$ by proving the following theorem:
\begin{theorem}
For the channel $(\mathds{D}_6,\mathcal{Y},W)$, let $X$ be a uniform random variable over the channel input and let the random variable $[X]$ indicate the coset of $\{1,x,x^2\}$ in $\mathds{D}_6$ where $X$ belongs to. i.e.
\begin{align*}
[X]=\left\{\begin{array}{ll}
\{1,x,x^2\}&\mbox{if }X\in \{1,x,x^2\}\\
\{y,xy,x^2y\}&\mbox{if }X\in \{y,xy,x^2y\}
\end{array}\right.
\end{align*}
Then the rate $R^*$ is achievable using pseudo-group codes over $\mathds{D}_6$ where
\begin{align*}
R^*\!=\min\left(\log_2 6 -H(X|Y),\frac{\log_2 6}{\log_2 3}\left[\log_2 3 -H(X|[X]Y)\right]\right)
\end{align*}
\end{theorem}
The rest of this section is devoted to give a sketch of the proof of this theorem. Consider the class of pseudo-group codes over $\mathds{D}_6$ of the form (\ref{eqn:encoder}) used for the channel $(\mathds{D}_6,\mathcal{Y},W)$. The set of messages is $\mathds{D}_6$ and for each message $u\in \mathds{D}_6^k$ the encoder maps it to $c\in \mathds{D}_6^n$ where $c=G(u)\cdot B$. At the receiver, after receiving the channel output $y\in\mathcal{Y}^n$, the decoder looks for a message $\hat{u}\in\mathds{D}_6^k$ such that $\hat{c}=G(\hat{u})\cdot B$ is jointly $\epsilon$-typical with $y$ with respect to $P_XW_{Y|X}$ where $P_X$ is uniform over $\mathds{D}_6$ and $\epsilon>0$ is arbitrary. If it finds a unique such $\hat{c}$, it decodes $y$ to $\hat{u}$, otherwise it declares error.\\
The expected value of the average probability of error for this coding scheme is given by
\begin{align*}
&\mathds{E}\{P_{avg}(err)\}=\sum_{u\in \mathds{D}_6^k}\!\frac{1}{6^k}\!\sum_{c\in \mathds{D}_6^n}\!\! P(G(u)\cdot B=c)\sum_{\tilde{u}\ne u} \sum_{y\in A_\epsilon^n(Y|c)}\\
&\qquad  \sum_{\tilde{c}\in A_\epsilon^n(X|y)} \!\!\!P(G(\tilde{u})\cdot B=\tilde{c}|G(u)\cdot B=c)W(y|c)+O(\epsilon)\\
\end{align*}

We need to evaluate the conditional probability $P(G(\tilde{u})\cdot B=\tilde{c}|G(u)\cdot B=c)$ to proceed. For $u,\tilde{u}\in \mathds{D}_6^k$ and $x,\tilde{x}\in\mathds{D}_6^n$, let $u=(u_1,\cdots,u_k)$ where $u_i=x^{a_i}y^{b_i}$ for $i=1,\cdots,k$ and $\tilde{u}=(\tilde{u}_1,\cdots,\tilde{u}_k)$ where $\tilde{u}_i=x^{\tilde{a}_i}y^{\tilde{b}_i}$ for $i=1,\cdots,k$. Also let $c=(c_1,\cdots,c_n)$ and $\tilde{c}=(\tilde{c}_1,\cdots,\tilde{c}_n)$ and define $\theta=c\tilde{c}^{-1}=(\theta_1,\cdots,\theta_n)$ where $\theta_i=x^{\alpha_i}y^{\beta_i}$. Define the following:
\begin{align*}
&N_1(c,\tilde{c})=\left\{i\in [1,\cdots,n]|\beta_i=1\right\}\\
&N_2(c,\tilde{c})=\left\{i\in [1,\cdots,n]|\beta_i=0, \alpha_i\ne 0\right\}\\
&N_3(c,\tilde{c})=\left\{i\in [1,\cdots,n]|\beta_i=0, \alpha_i= 0\right\}=n-n_1-n_2\\
&M_1(u,\tilde{u})=\left\{i\in [1,\cdots,k]|b_i\ne \tilde{b}_i\right\}\\
&M_2(u,\tilde{u})=\left\{i\in [1,\cdots,k]|b_i=\tilde{b}_i, a_i\ne \tilde{a}_i\right\}\\
&M_3(u,\tilde{u})=\left\{i\in [1,\cdots,k]|b_i=\tilde{b}_i, a_i=   \tilde{a}_i\right\} \!=\!k\!-\!m_1\!-m_2
\end{align*}
Also define $n_1(c,\tilde{c})=|N_1(c,\tilde{c})|$, $n_2(c,\tilde{c})=|N_2(c,\tilde{c})|$, $n_3(c,\tilde{c})=|N_3(c,\tilde{c})|$, $m_1(u,\tilde{u})=|M_1(u,\tilde{u})|$, $m_2(u,\tilde{u})=|M_2(u,\tilde{u})|$, $m_3(u,\tilde{u})=|M_3(u,\tilde{u})|$.

\begin{lemma}
For $u,\tilde{u}\in \mathds{D}_6^k$ and $c,\tilde{c}\in\mathds{D}_6^n$, if $\tilde{u}\ne u$, then
\begin{align*}
&P\left(G(\tilde{u})\cdot B=\tilde{c}|G(u)\cdot B=c\right)\\
&=\frac{1}{10^{kn}}\left[10^{k-m_1}\cdot 3\sum_{\substack{l=1\\l\mbox{ odd}}}^{m_1}{m_1\choose l} 9^{l-1}\right]^{n_1} \cdot\\
&\left[\frac{10^{k-m_1-m_2}(10^{m_2}+2)}{3}+10^{k-m_1}\cdot 3\sum_{\substack{l=2\\l\mbox{ even}}}^{m_1}\!\!\!{m_1\choose l}9^{l-1}\right]^{n_2} \cdot\\
&\left[\frac{10^{k-m_1-m_2}(10^{m_2}-1)}{3}+10^{k-m_1}\cdot 3\sum_{\substack{l=2\\l\mbox{ even}}}^{m_1}\!\!\!{m_1\choose l}9^{l-1}\right]^{n_3}
\end{align*}
Moreover, for a fixed $u$, let $T_{m_1,m_2}(u)$ be the set of all $\tilde{u}$ with $m_1(u,\tilde{u})=m_1$, $m_2(u,\tilde{u})=m_2$, then
\begin{align*}
\left|T_{m_1,m_2}(u)\right|&={k\choose m_1,m_2,m_3} \cdot 3^{m_1}\cdot 2^{m_2}\\
&={k\choose m_1} {k-m_1\choose m_2}\cdot 3^{m_1}\cdot 2^{m_2}
\end{align*}
\end{lemma}
\begin{proof}
First Note that
\begin{align*}
&P\left(G(\tilde{u})\cdot B=\tilde{c}|G(u)\cdot B=c\right)=\\
&P\left(G(\tilde{u})G^{-1}(u)=\theta|G(u)\cdot B=c\right)=\\
&P\left(G(\tilde{u})G^{-1}(u)=\theta\right)
\end{align*}
where the first equality follows since the multiplication operation is invertible and the second equality follows since $B$ is uniform and independent from other variables. We have
\begin{align*}
&P\left(G(\tilde{u})G^{-1}(u)=\theta\right)=\\
&\prod_{i=1}^n P\left(g_{i1}^{a_1}h_{i1}^{b_1}g_{i2}^{a_2}h_{i2}^{b_2}\cdots g_{ik}^{a_k}h_{ik}^{b_k}h_{ik}^{\tilde{b}_k}g_{ik}^{-\tilde{a}_k}\cdots h_{i1}^{\tilde{b}_1}g_{i1}^{-\tilde{a}_1}=\theta_i\right)
\end{align*}
Hence, we first find the $i$th probability in this expression for some arbitrary $i\in 1,2,\cdots,n$. Consider the case where $\beta_i=1$. Since the ensemble has a uniform distribution, we need to count the number of $g_{ij}$'s and $h_{ij}$'s such that the equality
\begin{align}\label{eqn:one_term}
g_{i1}^{a_1}h_{i1}^{b_1}g_{i2}^{a_2}h_{i2}^{b_2}\cdots g_{ik}^{a_k}h_{ik}^{b_k}h_{ik}^{\tilde{b}_k}g_{ik}^{-\tilde{a}_k}\cdots h_{i1}^{\tilde{b}_1}g_{i1}^{-\tilde{a}_1}=x^{\alpha_i}y^{\beta_i}=x^{\alpha_i}y
\end{align}
is satisfied and divide this number to the total number of choices. Use the equality $yx=x^2y$ to argue that the power of $y$ on the left hand side of this expression adds. i.e. the power of $y$ on the left hand side is equal to the sum of the powers of $y$ terms appearing in the expression. This is equal to
\begin{align*}
\sum_{\substack{j=1\\h_{ij}\in\{y,xy,x^2y\}}}^k (b_j+\tilde{b}_j)
\end{align*}
Where the addition is done mod-$2$. This can be written as
\begin{align*}
\sum_{\substack{j\in M_1(c,\tilde{c})\\h_{ij}\in\{y,xy,x^2y\}}} (b_j+\tilde{b}_j)=\left|\{j\in M_1|h_{ij}\in\{y,xy,x^2y\}\}\right|
\end{align*}
Let $L\subseteq M_1$ be the set of indices $j\in M_1$ where $h_{i,j}\in\{y,xy,x^2y\}$. Since the power of $y$ on the right hand side is equal to one, the cardinality of $L$ must be odd. We count the number of solutions of (\ref{eqn:one_term}) as follows: Let $L\subseteq M_1$ be arbitrary with an odd cardinality. For $j\notin M_1$ choose $g_{ij}$ and $h_{ij}$ arbitrarily ($10^{k-m_1}$ choices). for $j\in M_1\backslash L$ let $g_{ij}=1$ and $h_{ij}=1$ ($1$ choice). Since the cardinality of $L$ is assumed to be odd, it should have at least one element (say $j^*$). for $j\in L\backslash\{j^*\}$ choose $g_{ij}$ from $\{1,x,x^2\}$ and $h_{ij}$ from $\{y,xy,x^2y\}$ arbitrarily ($9^{l-1}$ choices where $l=|L|$). Also choose $g_{ij^*}$ arbitrarily from $\{1,x,x^2\}$ (3 choices). After moving terms to the other side we will get the following expression
\begin{align*}
h_{ij^*}=x^{\mbox{some power}}y^{\mbox{some power}}
\end{align*}
Note that, by construction, the power of $y$ on the right hand side has to be equal to $y$ and hence $h_{ij^*}$ has a unique solution in $\{y,xy,x^2y\}$. Therefore, it turns out that the number of choices for the case $\beta_i=1$ is equal to
\begin{align*}
10^{k-m_1}\cdot 3\cdot \sum_{\substack{l=1\\l\mbox{ odd}}}^{m_1} {m_1\choose l}9^{l-1}
\end{align*}
Note that if $m_1=0$ the above expression is defined to be zero.\\
Now consider the case where $\beta_i=0$. In this case it is possible to have $m_1=0$. We will first consider this case. Since $m_1=0$, for all $j=1,\cdots,k$ we have $\tilde{b}_j=b_j$. On the other hand, since $\tilde{u}\ne u$, there must exist an index $j\in 1,\cdots,k$ such that $\tilde{a}_j\ne a_j$. It can be shown that the total number of choices in this case is equal to
\begin{align*}
&10^{k-1}3+10^{k-2}3+10^{k-3}3+\cdots+10^{k-m_2}3\\
&=\frac{10^{k-m_2}(10^{m_2}-1)}{3}
\end{align*}
if $\theta_i\ne 1$ and it is equal to
\begin{align*}
&10^{k-1}3+10^{k-2}3+10^{k-3}3+\cdots+10^{k-m_2}4\\
&=\frac{10^{k-m_2}(10^{m_2}+1)}{3}
\end{align*}
if $\theta_i= 1$.
The next and the last case is where $\beta_i=0$ and $m_1\ne 0$. In this case the argument is similar to the case where $\beta_i=1$ (and hence $m_1\ne0$). The difference here is that we need to choose a subset $L$ of $M_1$ with an even cardinality $l$. For $l>0$ the argument is similar to the case with $\beta_i=1$ and for $l=0$, the argument is similar to the case with $\beta_i=0$ and $m_1=0$. Therefore the number of choices in this case is equal to
\begin{align*}
\frac{10^{k-m_2}(10^{m_2}-1)}{3}+\sum_{\substack{l=2\\l\mbox{ even}}}^{m_1}{m_1\choose l}10^{k-m_1}9^{l-1}3
\end{align*}
if $\theta_i\ne 1$ and it is equal to
\begin{align*}
\frac{10^{k-m_2}(10^{m_2}+2)}{3}+\sum_{\substack{l=2\\l\mbox{ even}}}^{m_1}{m_1\choose l}10^{k-m_1}9^{l-1}3
\end{align*}
if $\theta_i=1$.\\
Since the total number of choices is equal to $10^{kn}$ the assertion about the conditional probability in the lemma follows. For a fixed $u\in\mathds{D}_6^k$, the number of $\tilde{u}$'s in $\mathds{D}_6^k$ such that $m_1(u,\tilde{u})=m_1$ and $m_2(u,\tilde{u})=m_2$ is calculated as follows. Fix $m_1$ positions out of $k$ positions (${k\choose m_1}$ choices) and in these positions let $\tilde{b_j}=b_j+1$ and $\tilde{a}_j$ arbitrary ($3^{m_1}$ choices). Fix $m_2$ positions among the remaining $k-m_1$ positions (${k-m_1\choose m_2}$ choices) and in these positions let $\tilde{b}_j=b_j$ and choose $\tilde{a}_j\ne a_j$ ($2^{m_2}$ choices). It follows that the total number of choices for $\tilde{u}$ is equal to
\begin{align*}
\left|T_{m_1,m_2}(u)\right|&={k\choose m_1,m_2,m_3} \cdot 3^{m_1}\cdot 2^{m_2}\\
&={k\choose m_1} {k-m_1\choose m_2}\cdot 3^{m_1}\cdot 2^{m_2}
\end{align*}
\end{proof}

Define
\begin{align*}
&A(m_1)=\sum_{\substack{l=1\\l\mbox{ odd}}}^{m_1}\left(\begin{array}{c}m_1\\l\end{array}\right)9^{l}\\
&B(m_1,m_2)=\frac{(10^{m_2}+2)}{10^{m_2}}+\sum_{\substack{l=2\\l\mbox{ even}}}^{m_1}\left(\begin{array}{c}m_1\\l\end{array}\right)9^{l}\\
&C(m_1,m_2)=\frac{(10^{m_2}-1)}{10^{m_2}}+\sum_{\substack{l=2\\l\mbox{ even}}}^{m_1}\left(\begin{array}{c}m_1\\l\end{array}\right)9^{l}\\
\end{align*}
Using the above lemma and definitions, the expected value of the average probability of error can be upper bounded by:
\begin{align*}
&\mathds{E}\{P_{avg}(err)\}\\
&\le\sum_{m_1=0}^k \sum_{m_2=0}^{k-m_1} \sum_{n_1=0}^n \sum_{n_2=0}^{n-n_1} {k\choose m_1}{k-m_1 \choose m_2}\cdot 3^{m_1} \cdot 2^{m_2}\frac{1}{10^{kn}}\cdot\\
& 10^{n(k-m_1)}\cdot \frac{1}{3^n}A(m_1)^{n_1} B(m_1,m_2)^{n-n_1-n_2} C(m_1,m_2)^{n_2}\cdot \\
&\left|\left(x\cdot \{y,xy,x^2y\}^{n_1}\!\!\times\! \{x,x^2\}^{n_2}\!\!\times\! \{1\}^{n-n_1-n_2}\right)\cap A_\epsilon^n(X|y)\right|
\end{align*}
Define
\begin{align*}
&D(n_1,n_2;c,y)=\\
&\left|\left(c\cdot \{y,xy,x^2y\}^{n_1}\!\!\times\! \{x,x^2\}^{n_2}\!\!\times\! \{1\}^{n-n_1-n_2}\right)\cap A_\epsilon^n(X|y)\right|
\end{align*}
This quantity can be upper bounded by
\begin{align*}
\left|\left(c\cdot \{y,xy,x^2y\}^{n_1}\times \{1,x,x^2\}^{n-n_1}\right)\cap A_\epsilon^n(X|y)\right|
\end{align*}
and in turn we have the following lemma:

\begin{lemma}
Let $y\in \mathcal{Y}^n$ be an arbitrary channel output sequence. For any $x\in A_\epsilon^n(X|y)$, we have
\begin{multline*}
\left|\left(c\cdot \{y,xy,x^2y\}^{n_1}\times \{1,x,x^2\}^{n-n_1}\right)\cap A_\epsilon^n(X|y)\right|\\
\le {n \choose n_1} 2^{n\left[H(X|[X]Y)+O(\epsilon)\right]}
\end{multline*}
Where the random variable $[X]$ takes value from the set of cosets of $\{1,x,x^2\}$ in $\mathds{D}_6$.
\end{lemma}
\begin{proof}
First we prove the following:
\begin{align*}
\left|\left(c\cdot \{y,xy,x^2y\}^{n_1}\times \{1,x,x^2\}^{n-n_1}\right)\cap A_\epsilon^n(X)\right|\\
\le {n \choose n_1} 2^{n\left[H(X|[X])+O(\epsilon)\right]}
\end{align*}
The cardinality in question is related to the conditional entropy of a random variable $W$ jointly distributed with $X$ which satisfies the constraints of the following optimization problem:
\begin{align*}
&\min_{p(g,w)}-H(X,W)=\sum_{g\in\mathds{D}_6}\sum_{w\in\mathds{D}_6} p(g,w)\log p(g,w)\\
&\mbox{s.t.}\\
& \qquad \sum_{w\in\mathds{D}_6}p(g,w)=P_X(g)\\
& \qquad \sum_{w\in\mathds{D}_6}p(g\cdot w^{-1},w)=P_X(g)\\
& \qquad \sum_{g\in\mathds{D}_6}\sum_{w\in\{1,x,x^2\}}p(x,w)=\alpha\\
& \qquad \sum_{g\in\mathds{D}_6}\sum_{w\in\{y,xy,x^2y\}}p(x,w)=1-\alpha\\
\end{align*}
where $\alpha=\frac{n-n_1}{n}$ and the minimization is over all probability mass functions on $\mathds{D}_6\times \mathds{D}_6$. Let $p=P_X(\{1,x,x^2\})$. It can be shown that if $\alpha\ge |1-2p|$, the following distribution satisfies the KKT conditions for this optimization problem.
\begin{align*}
p_{XW}(x,w)=\left\{\begin{array}{ll}
\frac{2p-1+\alpha}{2p^2}p_X(x)p_X(x\cdot w)&x,w\in\{1,x,x^2\}\\
\frac{1-2p+\alpha}{2(1-p)^2}p_X(x)p_X(x\cdot w)&x\in\{y,xy,x^2y\}\\
&,w\in\{1,x,x^2\}\\
\frac{1-\alpha}{2p(1-p)}p_X(x)p_X(x\cdot w)&w\in\{y,xy,x^2y\}
\end{array}\right.
\end{align*}
The entropy of this joint pmf can be found to be equal to
\begin{align*}
H\!(\!X,\!W\!)\!&=\!2H(X)\!+\!h(\!\frac{2p-1+\alpha}{2}, \!\frac{1-2p+\alpha}{2},\!\frac{1-\alpha}{2},\!\frac{1-\alpha}{2}\!)\\
-2h(p)
\end{align*}
where $p=P_X(\{1,x,x^2\})$ and $h$ is the entropy function. Next we prove that for $\alpha\ge |1-2p|$,
\begin{align*}
h(\!\frac{2p-1+\alpha}{2}, \!\frac{1-2p+\alpha}{2},\!\frac{1-\alpha}{2},\!\frac{1-\alpha}{2}\!)\le h(\alpha)+h(p)
\end{align*}
For $\alpha=0$ this statement is trivial. Assume $\alpha\ne 0$. Note that
\begin{align*}
&h(\frac{2p-1+\alpha}{2},\frac{1-2p+\alpha}{2},\frac{1-\alpha}{2},\frac{1-\alpha}{2})=\\
&-\frac{2p-1+\alpha}{2}\log \frac{2p-1+\alpha}{2}-\frac{1-2p+\alpha}{2}\log \frac{1-2p+\alpha}{2}\\
&-(1-\alpha)\log \frac{1-\alpha}{2}\\
&=-\alpha\left[\frac{2p-1+\alpha}{2\alpha}\log \frac{2p-1+\alpha}{2\alpha}+\frac{1-2p+\alpha}{2\alpha}\right.\\
&\left.\log \frac{1-2p+\alpha}{2\alpha}+\log\alpha\right]-(1-\alpha)\log (1-\alpha)+(1-\alpha)\\
&=\alpha\left[h(\frac{2p-1+\alpha}{2\alpha})-h(\alpha)\right]-(1-\alpha)\log (1-\alpha)+(1-\alpha)\\
&=\alpha h(\frac{2p-1+\alpha}{2\alpha})+(1-\alpha)+h(\alpha)
\end{align*}
Since the function $h(\cdot)$ is convex, for two points $x_1,x_2\in[0,1]$ and a number $\alpha\in[0,1]$ we have
\begin{align*}
\alpha h(x_1)+(1-\alpha)h(x_2)\le h(\alpha x_1+(1-\alpha)x_2)
\end{align*}
Let $x_1=\frac{2p-1+\alpha}{2\alpha}$ and $x_2=\frac{1}{2}$ to get
\begin{align*}
\alpha h(\frac{2p-1+\alpha}{2\alpha})+(1-\alpha)\le h(\frac{2p-1+\alpha}{2}+\frac{1-\alpha}{2})=h(p)
\end{align*}
Therefore
\begin{align*}
h(\!\frac{2p-1+\alpha}{2}, \!\frac{1-2p+\alpha}{2},\!\frac{1-\alpha}{2},\!\frac{1-\alpha}{2}\!)\le h(\alpha)+h(p)
\end{align*}
Using the equality $H([X])=h(p)$ we get $H(W|X)\le H(X|[X])+h(\alpha)$. Finally we use Stirling's approximation to get
\begin{multline*}
\left|\left(c\cdot \{y,xy,x^2y\}^{n_1}\times \{1,x,x^2\}^{n-n_1}\right)\cap A_\epsilon^n(X)\right|\\
\le {n \choose n_1} 2^{n\left[H(X|[X])+O(\epsilon)\right]}
\end{multline*}
The generalization to the statement of the lemma is relatively straight forward and is omitted.
\end{proof}

It can be shown that for all $\delta>0$, there exists an integer $M_1(\delta)$ such that for $m_1\ge M_1(\delta)$, $A(m_1),B(m_1,m_2),C(m_1,m_2)<\frac{10^{m_1}}{2(1-\delta)}$. It can also be shown that for all $\delta^\prime>0$, there exists an integer $M_2(\delta^\prime)$ such that for $m_2\ge M_2(\delta^\prime)$, $A(m_1)<\frac{10^{m_1}-8^{m_1}}{2(1-\delta^\prime)}$ and $B(m_1,m_2),C(m_1,m_2)<\frac{10^{m_1}+8^{m_1}}{2(1-\delta^\prime)}$.\\
For arbitrary $\delta,\delta^\prime >0$, we break the expected error probability into several terms as follows:
\begin{align*}
\mathds{E}\{P_{avg}(err)\}=P_1(\delta)+\sum_{m_1=0}^{M_1(\delta)-1} P_2(m_1,\delta^\prime)+P_3(\delta,\delta^\prime)
\end{align*}
where
\begin{align*}
&P_1(\delta)\\
&=\sum_{m_1=M_1(\delta)}^k \sum_{m_2=0}^{k-m_1} \sum_{n_1=0}^n \sum_{n_2=0}^{n-n_1} {k\choose m_1} {k-m_1\choose m_2}\cdot 3^{m_1} \cdot 2^{m_2}\cdot\\
&\quad\frac{1}{10^{m_1n}}\cdot \frac{1}{3^n} A(m_1)^{n_1} B(m_1,m_2)^{n-n_1-n_2} C(m_1,m_2)^{n_2}\cdot\\
&\left|\left(c\cdot \{y,xy,x^2y\}^{n_1}\!\!\times\! \{x,x^2\}^{n_2}\!\!\times\! \{1\}^{n-n_1-n_2}\right)\cap A_\epsilon^n(X|y)\right|
\end{align*}
\begin{align*}
&P_2(m_1,\delta^\prime)=\\
&\sum_{m_2=M_2(\delta^\prime)}^{k-m_1} \sum_{n_1=0}^n \sum_{n_2=0}^{n-n_1} {k\choose m_1}{k-m_1\choose m_2} \cdot 3^{m_1} \cdot 2^{m_2}\\
& \frac{1}{10^{m_1n}}\cdot \frac{1}{3^n} A(m_1)^{n_1} B(m_1,m_2)^{n-n_1-n_2} C(m_1,m_2)^{n_2}\cdot \\
&\left|\left(c\cdot \{y,xy,x^2y\}^{n_1}\!\!\times\! \{x,x^2\}^{n_2}\!\!\times\! \{1\}^{n-n_1-n_2}\right)\cap A_\epsilon^n(X|y)\right|
\end{align*}
and $P_3(\delta,\delta^\prime)$ is defined similar to $P_1(\delta)$ except that the first summation runs from $0$ to $M_1(\delta)-1$ and the second summation runs from $0$ to $M_2(\delta^\prime)-1$.
Next we show that
\begin{align*}
P_1(\delta)&\le \exp_2\left\{-n\left[\log_2[6(1-\delta)]-\frac{k}{n}\log_2 6 -H(X|Y)\right]\right\}
\end{align*}
and
\begin{align*}
&P_2(m_1,\delta^\prime)\\
&\le\!\exp_2\! \left\{\! -\!n\!\left[\!\log_2\![3(1-\delta)]\!-\!\frac{k}{n}\log_2 3 -H(\!X\!|[\!X\!]Y\!)\!+\!O\!(\epsilon)\!\right]\!\right\}
\end{align*}
and note that $P_3(\delta,\delta^\prime)$ is independent of the rate $R$ (It can be shown that this term goes to zero as $n$ increases regardless of the value of $R$). We have
\begin{align*}
&P_1(\delta)\\
&=\sum_{m_1=M_1(\delta)}^k \sum_{m_2=0}^{k-m_1} \sum_{n_1=0}^n \sum_{n_2=0}^{n-n_1} {k\choose m_1} {k-m_1\choose m_2}\cdot 3^{m_1} \cdot 2^{m_2}\cdot\\
&\quad\frac{1}{10^{m_1n}}\cdot \frac{1}{3^n} A(m_1)^{n_1} B(m_1,m_2)^{n-n_1-n_2} C(m_1,m_2)^{n_2}\cdot\\
&\left|\left(c\cdot \{y,xy,x^2y\}^{n_1}\!\!\times\! \{x,x^2\}^{n_2}\!\!\times\! \{1\}^{n-n_1-n_2}\right)\cap A_\epsilon^n(X|y)\right|\\
&\le\sum_{m_1=M_1(\delta)}^k \sum_{m_2=0}^{k-m_1} \sum_{n_1=0}^n \sum_{n_2=0}^{n-n_1} {k\choose m_1} {k-m_1\choose m_2}\cdot 3^{m_1} \cdot 2^{m_2}\cdot\\
&\quad\frac{1}{10^{m_1n}}\cdot \frac{1}{3^n} (\frac{10^{m_1}}{2(1-\delta)})^n\cdot\\
&\left|\left(c\cdot \{y,xy,x^2y\}^{n_1}\!\!\times\! \{x,x^2\}^{n_2}\!\!\times\! \{1\}^{n-n_1-n_2}\right)\cap A_\epsilon^n(X|y)\right|\\
&\le\sum_{m_1=M_1(\delta)}^k \sum_{m_2=0}^{k-m_1}  {k\choose m_1} {k-m_1\choose m_2}\cdot 3^{m_1} \cdot 2^{m_2}\cdot (\frac{1}{6(1-\delta)})^n \cdot\\&
\sum_{n_1=0}^n \sum_{n_2=0}^{n-n_1}\\
&\left|\left(c\cdot \{y,xy,x^2y\}^{n_1}\!\!\times\! \{x,x^2\}^{n_2}\!\!\times\! \{1\}^{n-n_1-n_2}\right)\cap A_\epsilon^n(X|y)\right|
\end{align*}
Note that the result of the last two summations is simply equal to
\begin{align*}
\left|A_\epsilon^n(X|y)\right|=2^{n\left[H(X|Y)+O(\epsilon)\right]}
\end{align*}
Therefore,
\begin{align*}
&P_1(\delta)\le\sum_{m_1=M_1(\delta)}^k \sum_{m_2=0}^{k-m_1}  {k\choose m_1} {k-m_1\choose m_2}\cdot 3^{m_1} \cdot 2^{m_2}\cdot\\
&(\frac{1}{6(1-\delta)})^n \cdot 2^{n\left[H(X|Y)+O(\epsilon)\right]}\\
&\le (\frac{1}{6(1-\delta)})^n \cdot 2^{n\left[H(X|Y)+O(\epsilon)\right]}\\
&\sum_{m_1=0}^k {k\choose m_1} 3^{m_1} \sum_{m_2=0}^{k-m_1} {k-m_1\choose m_2} \cdot 2^{m_2}
\end{align*}
Note that the result of the last summation is equal to $3^{k-m_1}$ and hence the result of the last two summations is equal to $6^k$. Therefore,
\begin{align*}
P_1(\delta)&\le \exp_2\left\{-n\left[\log_2[6(1-\delta)]-\frac{k}{n}\log_2 6 -H(X|Y)\right]\right\}
\end{align*}
For a fixed $m_1$, consider
\begin{align*}
&P_2(m_1,\delta^\prime)=\\
&\sum_{m_2=M_2(\delta^\prime)}^{k-m_1} \sum_{n_1=0}^n \sum_{n_2=0}^{n-n_1} {k\choose m_1}{k-m_1\choose m_2} \cdot 3^{m_1} \cdot 2^{m_2}\\
& \frac{1}{10^{m_1n}}\cdot \frac{1}{3^n} A(m_1)^{n_1} B(m_1,m_2)^{n-n_1-n_2} C(m_1,m_2)^{n_2}\cdot \\
&\left|\left(c\cdot \{y,xy,x^2y\}^{n_1}\!\!\times\! \{x,x^2\}^{n_2}\!\!\times\! \{1\}^{n-n_1-n_2}\right)\cap A_\epsilon^n(X|y)\right|\\
&\le\sum_{m_2=M_2(\delta^\prime)}^{k-m_1} \sum_{n_1=0}^n \sum_{n_2=0}^{n-n_1} {k\choose m_1}{k-m_1\choose m_2} \cdot 3^{m_1} \cdot 2^{m_2}\\
& \frac{1}{10^{m_1n}}\cdot \frac{1}{3^n} \left[\frac{10^{m_1}-8^{m_1}}{2(1-\delta^\prime)}\right]^{n_1} \left[\frac{10^{m_1}+8^{m_1}}{2(1-\delta^\prime)}\right]^{n-n_1} \cdot \\
&\left|\left(c\cdot \{y,xy,x^2y\}^{n_1}\!\!\times\! \{x,x^2\}^{n_2}\!\!\times\! \{1\}^{n-n_1-n_2}\right)\cap A_\epsilon^n(X|y)\right|\\
&\le\sum_{m_2=M_2(\delta^\prime)}^{k-m_1} \sum_{n_1=0}^n {k\choose m_1}{k-m_1\choose m_2} \cdot 3^{m_1} \cdot 2^{m_2}\\
& \frac{1}{10^{m_1n}}\cdot \frac{1}{3^n} \left[\frac{10^{m_1}-8^{m_1}}{2(1-\delta^\prime)}\right]^{n_1} \left[\frac{10^{m_1}+8^{m_1}}{2(1-\delta^\prime)}\right]^{n-n_1} \cdot \\
&\left|\left(c\cdot \{y,xy,x^2y\}^{n_1}\!\!\times\! \{1,x,x^2\}^{n-n_1}\right)\cap A_\epsilon^n(X|y)\right|\\
&\le\sum_{m_2=M_2(\delta^\prime)}^{k-m_1} \sum_{n_1=0}^n {k\choose m_1}{k-m_1\choose m_2} \cdot 3^{m_1} \cdot 2^{m_2}\\
& \frac{1}{10^{m_1n}}\cdot \frac{1}{3^n} \left[\frac{10^{m_1}-8^{m_1}}{2(1-\delta^\prime)}\right]^{n_1} \left[\frac{10^{m_1}+8^{m_1}}{2(1-\delta^\prime)}\right]^{n-n_1} \cdot \\
&{n \choose n_1} 2^{n\left[H(X|[X]Y)+O(\epsilon)\right]}\\
&=\sum_{m_2=M_2(\delta^\prime)}^{k-m_1} {k\choose m_1}{k-m_1\choose m_2} \cdot 3^{m_1} \cdot 2^{m_2}\\
& \frac{1}{10^{m_1n}}\cdot \frac{1}{3^n} 2^{n\left[H(X|[X]Y)+O(\epsilon)\right]} \cdot \\
&\sum_{n_1=0}^n{n \choose n_1}\left[\frac{10^{m_1}-8^{m_1}}{2(1-\delta^\prime)}\right]^{n_1} \left[\frac{10^{m_1}+8^{m_1}}{2(1-\delta^\prime)}\right]^{n-n_1}\\
&=\sum_{m_2=M_2(\delta^\prime)}^{k-m_1} {k\choose m_1}{k-m_1\choose m_2} \cdot 3^{m_1} \cdot 2^{m_2}\\
& \frac{1}{10^{m_1n}}\cdot \frac{1}{3^n} 2^{n\left[H(X|[X]Y)+O(\epsilon)\right]} \cdot \\
&\left[\frac{10^{m_1}-8^{m_1}}{2(1-\delta^\prime)}+\frac{10^{m_1}+8^{m_1}}{2(1-\delta^\prime)}\right]^n\\
&\le \frac{1}{(3(1-\delta^\prime))^n} 2^{n\left[H(X|[X]Y)+O(\epsilon)\right]}\\
&{k\choose m_1} \cdot 3^{m_1} \sum_{m_2=0}^{k}{k\choose m_2} \cdot 2^{m_2}\\
&= \frac{1}{(3(1-\delta^\prime))^n} 2^{n\left[H(X|[X]Y)+O(\epsilon)\right]}{k\choose m_1} 3^{m_1} 3^k\\
\end{align*}
Note that for a fixed $m_1$, ${k\choose m_1} 3^{m_1}$ is a polynomial in $k$. Hence,
\begin{align*}
&P_2(m_1,\delta^\prime)\\
&\le\!\exp_2\! \left\{\! -\!n\!\left[\!\log_2\![3(1-\delta)]\!-\!\frac{k}{n}\log_2 3 -H(\!X\!|[\!X\!]Y\!)\!+\!O\!(\epsilon)\!\right]\!\right\}
\end{align*}
We observe that if $R<\log_2[6(1-\delta)]-H(X|Y)$ then $P_1(\delta)$ goes to zero as $n$ increases and if $R<\frac{\log_2 6}{\log_2 3}\left\{\log_2\left[3(1-\delta^\prime)\right]-H(X|[X]Y)\right\}$ then $\sum_{m_1=0}^{M_1(\delta)-1} P_2(m_1,\delta^\prime)$ vanishes as the block length increases. Therefore, If both conditions are satisfied the expected value of the block error probability goes to zero. In conclusion, The rate $R$ is achievable if
\begin{align*}
\left\{\begin{array}{l}
R<\log_2[6(1-\delta)]-H(X|Y)\\
R<\frac{\log_2 6}{\log_2 3}\left\{\log_2\left[3(1-\delta^\prime)\right]-H(X|[X]Y)\right\}
\end{array}\right.
\end{align*}
Since $\delta$ and $\delta^\prime$ are arbitrary, we conclude that the rate $R^*$ is achievable where
\begin{align*}
R^*=\min\left(\log_2 6 -H(X|Y),\frac{\log_2 6}{\log_2 3}\left[\log_2 3 -H(X|[X]Y)\right]\right)
\end{align*}
\section{Comparison With Abelian Group Codes}\label{section:comparison}
The only Abelian group of size $6$ is $\mathds{Z}_6=\{0,1,\cdots,5\}$ where the group operation is addition mod-$6$. The best achievale rate using abelian group codes over $\mathds{Z}_6$ is known to be \cite{abelianp2p_ieee}
\begin{align*}
R^*=\min\left(\log_2 6 -H(X|Y),\frac{\log_2 6}{\log_2 3}\left[\log_2 3 -H(X|[X]_3Y)\right]\right.&\\
\left.{\color{white}\frac{1}{2}},\log_2 6\left[1-H(X|[X]_2Y)\right]\right)&
\end{align*}
where $[X]_3$ takes values from cosets of $\{0,2,4\}$ and $[X]_2$ takes values from cosets of $\{0,3\}$. In the following example we show that the achievable rate using the new code can be strictly larger than the rate achievable using Abelian group codes.
\subsection{An Example}
We give an example where the capacity of group codes is zero whereas the constructed code achieves a strictly positive rate. Consider the channel depicted in the figure below where $\epsilon_1=0.1$, $\epsilon_1=0.2$ and $\epsilon_1=0.15$.
\begin{figure}[h!]\label{fig:channel}%
  \centering
    \includegraphics[scale=.7]{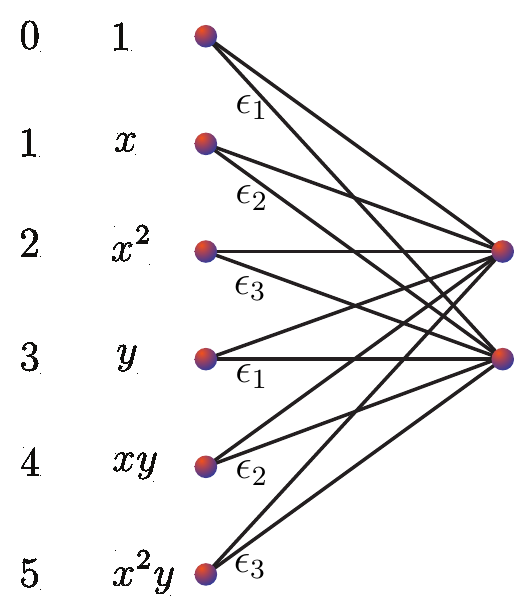}
    \caption{Channel: The first column on the left shows the input labels in $\mathds{Z}_6$ and the second column shows the labels in $\mathds{D}_6$.}
\end{figure}
If we maximize over all possible labellings of the channel input alphabet, it can be shown that both coding schemes achieve the symmetric capacity of the channel which is equal to $0.0139$ bits per channel use. However, if the labels are assumed to be fixed, the achievable rate using pseudo-group codes is equal to $R^*=\min(0.0139,0.0227)=0.0139$ and the achievable rate using Abelian group codes is equal to $R=\min(0.0139,0.0227,0)=0$.
Indeed using the converse provided in \cite{abelianp2p_ieee} we can show that the capacity of Abelian group codes over this channel is equal to zero. We observe that for this channel, the codes over $\mathds{D}_6$ outperforms the code over $\mathds{Z}_6$.
\subsection{Comparison}
If we compare the two achievable rate regions, we observe that for the case of abelian group codes there is an additional term in the minimization which can be explained by the additional structure of the abelian group codes. Indeed, the pseudo-group code over $\mathds{D}_6$ is additive (homomorphic) with respect to the $y$ generator and is not homomorphic with respect to the $x$ generator whereas Abelian group codes are homomorphic with respect to both of their generators. This means compared to Abelian Group codes, the constructed codes gain a higher rate by reducing the structure.
%
\section{Conclusion}\label{section:Conclusion}
 We have shown that good structured codes over non-Abelian groups do exist. We constructed codes over the smallest non-Abelian group $\mathds{D}_6$ and showed that the performance of these codes is superior to the performance of Abelian group codes of the same alphabet size.

\ifCLASSOPTIONcaptionsoff
  \newpage
\fi

\bibliographystyle{IEEEtran}
\bibliography{IEEEabrv,ariabib}

\begin{thebibliography}{10}
\providecommand{\url}[1]{#1}
\csname url@samestyle\endcsname
\providecommand{\newblock}{\relax}
\providecommand{\bibinfo}[2]{#2}
\providecommand{\BIBentrySTDinterwordspacing}{\spaceskip=0pt\relax}
\providecommand{\BIBentryALTinterwordstretchfactor}{4}
\providecommand{\BIBentryALTinterwordspacing}{\spaceskip=\fontdimen2\font plus
\BIBentryALTinterwordstretchfactor\fontdimen3\font minus
  \fontdimen4\font\relax}
\providecommand{\BIBforeignlanguage}[2]{{%
\expandafter\ifx\csname l@#1\endcsname\relax
\typeout{** WARNING: IEEEtran.bst: No hyphenation pattern has been}%
\typeout{** loaded for the language `#1'. Using the pattern for}%
\typeout{** the default language instead.}%
\else
\language=\csname l@#1\endcsname
\fi
#2}}
\providecommand{\BIBdecl}{\relax}
\BIBdecl

\bibitem{ahlswede_alg_codes}
R.~Ahlswede and J.~Gemma, ``{B}ounds on algebraic code capacities for noisy
  channels {I},'' \emph{Information and Control}, vol.~19, no.~2, pp. 124--145,
  1971.

\bibitem{sandeep_discus}
S.~S. Pradhan and K.~Ramchandran, ``{D}istributed source coding using syndromes
  ({DISCUS}): Design and construction,'' \emph{{IEEE} Transactions on
  Information Theory}, vol.~49, no.~3, pp. 626--643, 2003.

\bibitem{forney_dynamics}
G.~D.~F. Jr and M.~Trott, ``{T}he dynamics of group codes: State spaces,
  trellis diagrams, and canonical encoders,'' \emph{{IEEE} Transactions on
  Information Theory}, vol.~39, no.~9, pp. 1491--1513, 1993.

\bibitem{fagnani_abelian}
G.~Como and F.~Fagnani, ``{T}he capacity of finite abelian group codes over
  symmetric memoryless channels,'' \emph{{IEEE} Transactions on Information
  Theory}, vol.~55, no.~5, pp. 2037--2054, 2009.

\bibitem{dinesh_dsc}
D.~Krithivasan and S.~S. Pradhan, ``{D}istributed source coding using abelian
  group codes,'' 2011, {IEEE} Transactions on Information Theory(57)1495-1519.

\bibitem{dobrushin_group}
R.~L. Dobrushin, ``{A}symptotic optimality of group and systematic codes for
  some channels,'' \emph{Theor. Probab. Appl.}, vol.~8, pp. 47--59, 1963.

\bibitem{elias}
P.~Elias, ````{C}oding for noisy channels'','' \emph{{IRE} {C}onv. {R}ecord},
  vol. part. 4, pp. 37--46, 1955.

\bibitem{korner_marton}
J.~Korner and K.~Marton, ``{H}ow to encode the modulo-two sum of binary
  sources,'' \emph{{IEEE} Transactions on Information Theory}, vol. IT-25, pp.
  219--221, Mar. 1979.

\bibitem{phiosof_zamir}
T.~Philosof, A.~Kishty, U.~Erez, and R.~Zamir, ``{L}attice strategies for the
  dirty multiple access channel,'' \emph{{P}roceedings of {IEEE} International
  Symposium on Information Theory}, July 2007, nice, France.

\bibitem{nazer_gastpar}
B.~A. Nazer and M.~Gastpar, ``{C}omputation over multiple-access channels,''
  \emph{{IEEE} Transactions on Information Theory}, vol.~53, no. 10 pages =,
  Oct. 2007.

\bibitem{slepian_group}
D.~Slepian, ``{G}roup codes for for the {G}aussian channel,'' \emph{{B}ell
  {S}yst. {T}ech. {J}ournal}, 1968.

\bibitem{ahlswede_group}
R.~Ahlswede, ``{G}roup codes do not achieve {S}hannons's channel capacity for
  general discrete channels,'' \emph{The annals of Mathematical Statistics},
  vol.~42, no.~1, pp. 224--240, Feb. 1971.

\bibitem{ahlswede_alg_codes2}
R.~Ahlswede and J.~Gemma, ``{B}ounds on algebraic code capacities for noisy
  channels {II},'' \emph{Information and Control}, vol.~19, no.~2, pp.
  146--158, 1971.

\bibitem{abelianp2p_ieee}
A.~G. Sahebi and S.~S. Pradhan, ``{O}n the {C}apacity of {A}belian {G}roup
  {C}odes {O}ver {D}iscrete {M}emoryless {C}hannels,'' July 2011, petersburg,
  Russia.

\bibitem{Forney_hamming_distance_group_codes}
D.~Forney, ``{O}n the {H}amming {D}istance {P}roperties of {G}roup {C}odes.''

\bibitem{Interlando_group_codes_are_bad}
J.~Interlando, R.~Palazzo, and M.~Elia, ``{G}roup block codes over nonabelian
  groups are asymptotically bad,'' \emph{{IEEE} {T}ransactions on {I}nformation
  {T}heory}, vol.~42, pp. 1277--1280, 1996.

\bibitem{Massey_nonabelian_conformant}
P.~Massey, ``{M}any {N}on-{A}belian {G}roups {S}upport {O}nly {G}roup {C}odes
  {T}hat {A}re {C}onformant {T}o {A}belian {G}roup {C}odes.''

\bibitem{algebra_bloch}
N.~J. Bloch, \emph{{A}bstract Algebra With Applications}.\hskip 1em plus 0.5em
  minus 0.4em\relax Englewood Cliffs, New Jersey: Prentice-Hall, Inc, 1987.

\end{thebibliography}
\end{document}